\documentclass[nofootinbib,groupedaddress,superscriptaddress,aps,prl,reprint]{revtex4-2}

\let\coloneqq\relax
\let\eqqcolon\relax

\usepackage[latin1]{inputenc}
\usepackage{amsmath, amssymb, amsthm}
\usepackage{mathtools}
\usepackage[dvipsnames]{xcolor}
\usepackage[pdftex]{hyperref}
\definecolor{apsblue}{HTML}{2E3092}
\hypersetup{
    colorlinks=true,
    linkcolor=apsblue,    
    citecolor=apsblue,    
    urlcolor=apsblue    
}
\usepackage{braket}
\usepackage[shortlabels]{enumitem}
\usepackage[cal=boondox]{mathalfa}
\usepackage{graphicx}
\usepackage{stackengine}
\usepackage{comment}
\usepackage{sansmath}

\usepackage{tikz}
\usetikzlibrary{calc}

\newtheorem{thm}{Theorem}
\newtheorem{prop}[thm]{Proposition}
\newtheorem{lemma}[thm]{Lemma}
\newtheorem{cor}[thm]{Corollary}
\newtheorem{Def}[thm]{Definition}

\makeatletter
\def\thmhead@plain#1#2#3{%
  \thmname{#1}\thmnumber{\@ifnotempty{#1}{ }\@upn{#2}}%
  \thmnote{ {\the\thm@notefont#3}}}
\let\thmhead\thmhead@plain
\makeatother

\theoremstyle{definition}
\newtheorem{rem}[thm]{Remark}

\newcommand{\bb}{\begin{equation}\begin{aligned}\hspace{0pt}}
\newcommand{\ee}{\end{aligned}\end{equation}}
\newcommand*{\coloneqq}{\mathrel{\vcenter{\baselineskip0.5ex \lineskiplimit0pt \hbox{\scriptsize.}\hbox{\scriptsize.}}} =}
\newcommand*{\eqqcolon}{= \mathrel{\vcenter{\baselineskip0.5ex \lineskiplimit0pt \hbox{\scriptsize.}\hbox{\scriptsize.}}}}
\newcommand{\eqt}[1]{\stackrel{\mathclap{\scriptsize \mbox{#1}}}{=}}
\newcommand{\ketbra}[1]{\ket{#1}\!\!\bra{#1}}

\DeclareMathAlphabet{\pazocal}{OMS}{zplm}{m}{n}
\newcommand{\R}{\mathbb{R}}
\newcommand{\MM}{\pazocal{M}}

\DeclareMathOperator{\tr}{tr}

\stackMath

\stackMath
\newcommand\xxlrightarrow[2][]{\mathrel{%
  \setbox2=\hbox{\stackon{\scriptstyle#1}{\scriptstyle#2}}%
  \stackunder[1pt]{%
    \xrightarrow{\makebox[\dimexpr\wd2\relax]{$\scriptstyle#2$}}%
  }{%
   \scriptstyle#1\,%
  }%
}}

\newcommand{\tendsnl}[1]{\xxlrightarrow[\! n\rightarrow \infty\!]{#1}}

\newcommand{\G}{{\cal G}}

\usepackage{xfrac}

\usepackage{newtxtext,txfonts}
\usepackage{arimo}
\usepackage{sansmath}

\usepackage{dsfont}
\DeclareSymbolFont{dsromsym}{U}{dsrom}{m}{n}
\SetSymbolFont{dsromsym}{sans}{U}{dsrom}{m}{n}
\DeclareMathSymbol{\id}{\mathord}{dsromsym}{`1}

\begin{document}

\title{Optimal Discrimination of Gaussian States by Gaussian Measurements}

\author{Leah Turner}
\email{pmylt1@nottingham.ac.uk}
\affiliation{School of Mathematical Sciences and Centre for the Mathematical and Theoretical Physics of Quantum Non-Equilibrium Systems, University of Nottingham, University Park, Nottingham, NG7 2RD, United Kingdom}

\author{Ludovico Lami}
\email{ludovico.lami@sns.it}
\affiliation{Scuola Normale Superiore, Piazza dei Cavalieri 7, 56126 Pisa, Italy}

\author{M\u{a}d\u{a}lin Gu\c{t}\u{a}}
\email{madalin.guta@nottingham.ac.uk}
\affiliation{School of Mathematical Sciences and Centre for the Mathematical and Theoretical Physics of Quantum Non-Equilibrium Systems, University of Nottingham, University Park, Nottingham, NG7 2RD, United Kingdom}

\author{Gerardo Adesso}
\email{gerardo.adesso@nottingham.ac.uk}
\affiliation{School of Mathematical Sciences and Centre for the Mathematical and Theoretical Physics of Quantum Non-Equilibrium Systems, University of Nottingham, University Park, Nottingham, NG7 2RD, United Kingdom}

\begin{abstract} 
   Are Gaussian measurements enough to distinguish between Gaussian states? Here, we tackle this question by focusing on the max-relative entropy as an operational distinguishability metric. Given two general multimode Gaussian states, we derive a condition, based on their covariance matrices, that completely determines whether or not there exists an optimal Gaussian measurement achieving the max-relative entropy. When the condition is satisfied, we find this optimal measurement explicitly. When the condition is not met, there is a strict gap between the distinguishability achievable by Gaussian measurements and the unconstrained max-relative entropy in which all measurements are allowed. We illustrate our results in the single-mode setting, and show examples of states for which this gap can be made arbitrarily large, revealing novel instances of Gaussian data hiding. 
\end{abstract}

\maketitle

\paragraph{\bfseries Introduction.} 
Quantum state discrimination underpins many tasks in quantum information, from metrology to communication and cryptography~\cite{disco1,disco2,disco3}. In principle, any valid quantum measurement may be used to distinguish states. In practice, physical implementations typically allow only restricted sets of operations, motivating the study of information-theoretic quantities under operational constraints. Well-known examples include restrictions to local operations and classical communication~\cite{locc1,locc2,locc3,locc4,locc5,Leone_2025,Chefles_2004,Fan_2004,PhysRevA.88.052313} --- a typical paradigm in distributed scenarios --- and recently studied limitations arising from bounded computational complexity~\cite{computational1,computational2,computational3,computational4,avidan2025fullyquantumcomputationalentropies,chen2017computationalnotionsquantumminentropy,Leone_2025}, which provide a more realistic platform for quantum information processing. 

A third example, and the focus of this Letter, is the restriction to {\em Gaussian states, operations, and measurements}, that is particularly natural in continuous variable quantum information~\cite{gaussian1,Nha_2005,Oh_2019no2,PhysRevA.78.022320,Roberson_2021,Oh_2019,gaussian2,gaussian3}. These processes admit an elegant phase-space description and can be implemented experimentally using linear optics, squeezing, and homodyne or heterodyne detection, making them central to photonic quantum technologies~\cite{BUCCO}. This raises a fundamental question: can {\it only} Gaussian measurements suffice to discriminate between Gaussian states to the same standard as {\it arbitrary} quantum measurements?

Surprisingly, there are very limited known cases in which the answer is affirmative. {Gaussian measurements are optimal for estimating the fidelity between pure undisplaced single-mode Gaussian states~\cite{Nha_2005}, while for mixed states the optimal measurement is Gaussian only in restricted cases ~\cite{Oh_2019no2}.} For phase estimation with single-mode Gaussian states, Gaussian measurements are optimal only when the initial state is a displaced thermal state or squeezed vacuum~\cite{Oh_2019}. { Again in the context of parameter estimation, some sufficient conditions under which Gaussian measurements achieve the quantum Fisher information are known~\cite{monras2013phasespaceformalismquantum}, including when the state is pure. Amongst the class of Gaussian measurements,} homodyne measurements optimally discriminate between two single-mode coherent states~\cite{PhysRevA.78.022320}, or between mixtures of coherent states distributed along a line in phase space~\cite{Roberson_2021}. The sub-optimality of Gaussian measurements  can lead to a phenomenon known as {\em Gaussian data hiding}, where states almost perfectly distinguishable by general measurements are nearly indistinguishable via Gaussian measurements alone~\cite{PhysRevA.104.052428,sabapathy2021bosonicdatahidingpower,wang2025gaussianquantumdatahiding}.

In this Letter we address the question above for a key distinguishability measure: the {\it max-relative entropy}~\cite{min-maxRE,datta2009max}. This quantity, which sits at the upper end of all sandwiched R\'enyi divergences, plays an important role in quantum hypothesis testing and prominently features in quantum resource theories, where it is used to define operational resource monotones~\cite{min-maxRE,datta2009max,Mosonyi_2014,Croke_2006,postselected,ji2025hoeffdingchernofftradingconclusiveness,Napoli2016,Bu_2017,finiteconvexrobustness}. Crucially, the max-relative entropy $D_{\max}(\rho\|\sigma)$ is achievable by an optimal measurement on the states $\rho$ and $\sigma$, provided all  measurements are allowed~\cite{bertasquallor}. Our question thus becomes: \emph{when can the max-relative entropy between two Gaussian states be achieved by only Gaussian measurements?}

We answer this question completely. Unlike previously known cases, we find that Gaussian measurements can actually be optimal for a wide range of Gaussian discrimination problems gauged by the max-relative entropy. For arbitrary multimode Gaussian states $\rho$ and $\sigma$, we derive an exact covariance matrix condition that determines whether Gaussian measurements can or cannot attain $D_{\max}(\rho\|\sigma)$. When the condition is satisfied, we construct the optimal Gaussian measurement, which reduces to a homodyne-type detection at the boundary of the achievable region; outside this region, a finite gap emerges between the Gaussian-restricted and the unrestricted max-relative entropy. 
We showcase our results through the physically insightful case of single-mode states, where explicit examples demonstrate that Gaussian  measurements can, in fact, be significantly worse than arbitrary ones, revealing a novel {operational incarnation of Gaussian data hiding}.

\label{sec:prelim}
%This Section provides a brief overview of relevant material for this paper, covering Gaussian quantum information, and the max-relative entropy. Throughout, $\|\cdot\|_\infty$ denotes the operator norm, $\ddh$ is the set of quantum states, and $\ln$ denotes logarithm base {\rm e}.
%We begin by providing a succinct overview on Gaussian quantum information and the max-relative entropy. Throughout this paper, $\|\cdot\|_\infty$ denotes the operator norm, $\ddh$ is the set of quantum states, and $\ln$ denotes logarithm base {\rm e}.

\paragraph{\bfseries Gaussian states and measurements.} A continuous variable quantum system can be described as a system of $N$ modes, each of which has a corresponding position quadrature operator $q$ and momentum quadrature operator $p$, collected into a vector $x\coloneqq(q_1,p_1,...,q_N,p_N)$. 
The commutation relations between the quadrature operators are encoded in the symplectic form 
%\begin{equation}\label{eq:symplectic}    
$\Omega\coloneqq\bigoplus_{j=1}^N{{\ 0\ \ 1}\choose{-1\, \ 0}}$,
%\end{equation}
via  $[x_j,x_k]=i\Omega_{jk}$.
An $N$-mode continuous variable quantum state $\rho$ can be described by its Wigner function~\cite{wignerfunc}, a quasiprobability distribution of the quadrature operators over the phase space.
%, defined as 
%\begin{equation}\label{wig}
%    $W^\rho(r) = \pi^{-2N} \int_{\mathbb{R}^{2N}} \tr[{\rm e}^{-ir^T \Omega x}\rho] {\rm e}^{is^T\Omega r} d^{2N} s$.
%\end{equation}
Those states $\rho$ for which the Wigner function $W^\rho(x)$ is Gaussian, i.e.\ of the form 
\begin{equation}
    W^\rho(x) = \frac{1}{\pi^N\sqrt{\det(V)}}{\rm e}^{{-(x-r)^T V^{-1}(x-r)}},
\end{equation}
form the set of Gaussian states, and arise naturally as thermal and ground states of quadratic Hamiltonians~\cite{CVQI_and_beyond,BUCCO}. Such states are uniquely described by their first moments $r\coloneqq\tr[\rho x]$ and second moments $V$, $V_{jk}\coloneqq\tr[\rho\{x_j-r_j,x_k-r_k\}]$. A symmetric matrix $V$ describes a valid quantum covariance matrix if and only if it satisfies the {\it bona fide} condition~\cite{PhysRevA.49.1567}
\begin{equation}\label{eq:bonafide}
    V\pm i\Omega \geq 0.
\end{equation}

A Gaussian operation is a completely positive map that will preserve the Gaussian nature of a state. Each Gaussian unitary transformation can be described in phase space by some symplectic matrix $S$, i.e.\ a matrix preserving the symplectic form, $S\Omega S^T=\Omega$.
A general Gaussian unitary operation will consist of combinations of displacements, beam splitters, rotations, and squeezing, all of which are relatively available with current optical technologies.
By a symplectic transformation, any legitimate quantum covariance matrix $V$ can be diagonalized via its Williamson decomposition~\cite{williamson},
\begin{equation}\label{williamson}
    V=\mbox{$S^T\left(\bigoplus_{j=1}^N D_j\right) S$},
\end{equation}
where $D_j = V_\rho^{(j)} \id_2$, with $\{V_\rho^{(j)}\}_j$ the symplectic eigenvalues. This amounts to transforming to a basis in which the corresponding state $\rho$ is a product of single-mode thermal states.

A Gaussian measurement is a measurement that produces a Gaussian probability distribution for any Gaussian input state. Its action is specified by a Gaussian {\it seed} state with covariance matrix $\gamma$ { and first moments $d$}, so that acting on a Gaussian state with first and second moments $r$ and $V$, it produces a classical Gaussian distribution with mean $r${$-d$} and covariance $V+\gamma$. Pure measurements satisfy $\det\gamma=1$, while noisy measurements correspond to $\det\gamma>1$~\cite{noGdistillation3}. The limit of an infinitely squeezed seed yields a singular measurement, corresponding in the single-mode case to homodyne detection of some quadrature. In practice, any Gaussian measurement can be realized by adding Gaussian ancillary systems, performing a global Gaussian unitary, implementing homodyne detections, and discarding subsystems~\cite{noGdistillation3}.

\begin{comment}
In the rest of this work we shall make use of a few known results for Gaussian states, which we will now state.
\begin{lemma}[\cite{petzrecovery,grenyi2018}]\label{lemma:covmat}
    For two zero mean Gaussian quantum states $X,Y$, $X>Y$, the operators $X^{\pm\frac12}YX^{\pm\frac12}$ are trace-class Gaussian operators, and have covariance matrices 
    \begin{equation}
\pm V_X-\left( \sqrt{\id+(V_X \Omega)^{-2}}\right) V_X (V_X \pm V_Y)^{-1}V_X \left( \sqrt{\id+(\Omega V_X )^{-2}}\right).
    \end{equation}
\end{lemma}

We will also use a slight variation of~\cite[Lemma 10]{grenyi2018}, as stated below:
\begin{lemma}\label{collapse}
For a Gaussian quantum state $X$, we have 
%\begin{equation}
    $V_X-i\Omega=\left( \sqrt{\id+(V_X \Omega)^{-2}}\right) V_X (V_X+i\Omega)^{-1}V_X \left( \sqrt{\id+(\Omega V_X )^{-2}}\right)$. 
%    \end{equation}
\end{lemma}
\end{comment}

\paragraph{\bfseries Max-relative entropy.}
The quantum \emph{max-relative entropy} between two states $\rho$ and $\sigma$ is defined as~\cite{min-maxRE,datta2009max}
\begin{equation}\label{dmaxdef}
    D_{\max}(\rho\|\sigma)\coloneqq
    \begin{cases}
\inf\,\{\lambda\in\mathbb{R}_{\geq 0}:\rho\leq {\rm e}^\lambda \sigma\},& \text{supp}(\rho)\subseteq\text{supp}(\sigma)\\
\infty, & \text{otherwise}.
    \end{cases}
\end{equation}
Two equivalent and instructive formulations for the max-relative entropy are written in terms of the operator norm,
\begin{equation}\label{opnormdef}
    D_{\max}(\rho\|\sigma)\coloneqq\ln\|\sigma^{-\frac12}\rho\sigma^{-\frac12}\|_\infty,
\end{equation}\label{measureddef}
and via the {\em  measured} max-relative entropy,
\begin{equation}\label{dmaxmeas}
  \!\! D_{\max}(\rho\|\sigma)\coloneqq\ln \sup\left\{\frac{\tr[\rho E]}{\tr[\sigma E]}:E \text{ is a POVM element}\right\}.
\end{equation}
The equivalence between~\eqref{dmaxmeas} and~\eqref{dmaxdef} means that one can, in general, achieve $D_{\max}(\rho\|\sigma)$ by performing an optimal quantum measurement on $\rho$ and $\sigma$~\cite{bertasquallor}; in fact, the supremum in~\eqref{dmaxmeas} can be restricted to projective measurements~\cite[Corollary~3.46]{strongconverse}.

The max-relative entropy plays a prominent operational role {in quantum discrimination. Given an ensemble $\{p_k,\rho_k\}_{k=1}^m$, in which a state $\rho_k$ is prepared with probability $p_k$, the minimum Bayesian error probability for identifying the prepared state is $1-\inf_{\sigma} \max_{1\leq k\leq m}p_k {\rm e}^{D_{\max}(\rho_k\|\sigma)}$~\cite{Konig2009, min-maxRE}. Importantly, $D_{\max}$ admits a direct interpretation in \emph{maximum-confidence discrimination}~\cite{Croke_2006}. For two hypotheses $\rho$ and $\sigma$, with priors $p$ and $1-p$, respectively, the largest posterior confidence with which a conclusive measurement outcome can identify $\rho$ is
\begin{equation}\label{eq:maxconfidence}
C_\rho^\star(p)=
\frac{p\,{\rm e}^{D_{\max}(\rho\|\sigma)}}
{1-p+p\,{\rm e}^{D_{\max}(\rho\|\sigma)}}.
\end{equation}
Equivalently, ${\rm e}^{D_{\max}(\rho\|\sigma)}$ is the largest likelihood ratio obtainable in favour of $\rho$ from any measurement event. Beyond these state discrimination settings, $D_{\max}$ is directly relevant in quantum hypothesis testing. 
In the strong-converse regime, where the type II error probability is required to decay faster than the relative entropy rate, $D_{\max}$ precisely gives the asymptotic decay rate of the ratio between the type II error probability and the type~I \emph{success} probability~\cite{Mosonyi_2014}.
When inconclusive outcomes are allowed, max-relative entropies determine the asymptotic symmetric and asymmetric error exponents in postselected hypothesis testing~\cite{postselected}, and generally quantify trade-offs between conclusiveness and discrimination performance~\cite{ji2025hoeffdingchernofftradingconclusiveness}.
The max-relative entropy also enjoys operational significance in {\em quantum resource theories}~\cite{coherence,RT-review}: it defines the max-relative entropy of resource~\cite{Napoli2016,Bu_2017,min-maxRE} and is closely related, upon exponentiation, to the generalized robustness~\cite{datta2009max}, which quantifies universal advantages in channel discrimination~\cite{finiteconvexrobustness,finitenonconvex,infiniteconvex,catchytitle}.}

For two Gaussian states $\rho,\sigma$ with zero displacements and covariance matrices $V_\rho,V_\sigma$ respectively, satisfying $V_\sigma>V_\rho$, the max-relative entropy between $\rho$ and $\sigma$, when the measurements are unrestricted, can be found via~\cite[Theorem~24]{grenyi2018}
\begin{equation}\label{gdmaxformula}
%\begin{aligned}
     D_{\max}(\rho\|\sigma)=\frac12\!\ln\mbox{$\left[\dfrac{\det(V_\sigma+i\Omega)}{\det(V_\rho+i\Omega)}\right]$} 
     -\frac12\!\tr\left[\text{arcoth}\big(\!\sqrt{-V'\Omega V'\Omega}\big)\right]\!,
%\end{aligned}
\end{equation}
where $V' \coloneqq V_\rho+G(V_\rho \Omega) V_\rho (V_\sigma - V_\rho)^{-1}V_\rho G(\Omega V_\rho)$, with
%\begin{equation}
$G(V)\coloneqq \sqrt{\id+V^{-2}}$.
%\end{equation}
%We shall use~\eqref{gdmaxformula} to calculate the measured max-relative entropy in the case where the measurements are unrestricted, and explicitly compare this to the measured max-relative entropy when restricted to only Gaussian measurements.

\paragraph{\bfseries Gaussian measured max-relative entropy.}
In this Letter, we are concerned with determining how well Gaussian measurements alone can distinguish between Gaussian states. 
Recall that the classical max-relative entropy between two multivariate Gaussian distributions $P(A,a)$ and $Q(B,b)$, with $B>A$,  is given by~\cite{GClassicalRenyi} $D_{\max}(P\|Q)=\frac12\ln\left(\frac{\det B}{\det A}\right)+\frac12(a-b)^T(B-A)^{-1}(a-b).$
Substituting $P$ and $Q$ to be the post-measurement distributions resulting from a Gaussian measurement { described by a covariance matrix $\gamma$ and first moments $d$} implemented on two Gaussian quantum states $\rho$ and $\sigma$, respectively, and optimizing the ensuing expression over all possible choices of {the measurement seed}, we define the \emph{Gaussian measured max-relative entropy}, our main quantity of interest.
\begin{Def}\label{def:gdmax}
Let $\rho$ and $\sigma$ be Gaussian states with first moments $r_\rho$ and $r_\sigma$, and second moments $V_\rho$ and $V_\sigma$ respectively, with $V_\sigma>V_\rho$. The Gaussian measured max-relative entropy is
    \begin{equation}\label{eq:gdmax}
\begin{aligned}
    D^{\cal G}_{\max}(\rho\|\sigma)&\coloneqq \sup_{\gamma\geq i\Omega}\frac12\ln\left[\frac{\det (V_\sigma+\gamma)}{\det (V_\rho+\gamma)}\right] \\
    &+\frac12(r_\rho-r_\sigma)^T(V_\sigma-V_\rho)^{-1}(r_\rho-r_\sigma).
    \end{aligned}
\end{equation}
\end{Def}

Note that all dependence on first moments is in the second-line term, which is independent of  $\gamma$. We will thus restrict without loss of generality to Gaussian states with $r_\rho=r_\sigma=0$, { and measurements with $d=0$}, as displacements play no role in the measurement optimization. We will assume throughout that  $\sigma>0$ and $V_\sigma>V_\rho$, so we can compare our quantity~\eqref{eq:gdmax} to the max-relative entropy~\eqref{gdmaxformula} when the measurements are unrestricted, without divergence issues { for both $D_{\max}$~\cite[Prop.~25]{grenyi2018} and $D^{\cal G}_{\max}$}. Furthermore, we can limit the maximization in Definition~\ref{def:gdmax} to {\it pure} Gaussian measurements only ($\det\gamma=1$)~\cite{notepuremeasurement}\nocite{Lami_2016}.

\paragraph{\bfseries Optimality of Gaussian measurements.}

We now address the central question: when are Gaussian measurements sufficient to achieve the max-relative entropy~\eqref{dmaxdef} between two (undisplaced~\cite{notedisplacement}) $N$-mode Gaussian states?
%, i.e.\ when:
\begin{equation}
D^{\cal G}_{\max}(\rho\|\sigma)
%= \sup_{\gamma\geq i\Omega}\frac12\ln\left(\frac{\det (V_\sigma+\gamma)}{\det (V_\rho+\gamma)}\right)
\eqt{?}D_{\max}(\rho\|\sigma)\, .
\label{equality?}   
\end{equation}

To this end, consider the likelihood operator $M\coloneqq\sigma^{-\frac12}\rho\sigma^{-\frac12}$ in the definition~\eqref{opnormdef} of $D_{\max}$. This is  Gaussian and trace-class, with first moments $r_M=0$ and covariance matrix~\cite{petzrecovery,grenyi2018,lemmings1}
\begin{equation}\label{vm}
    V_M=-V_\sigma-G(V_\sigma \Omega) V_\sigma (V_\rho -V_\sigma)^{-1}V_\sigma G(\Omega V_\sigma ).
\end{equation}
To find an optimal measurement for $D_{\max}$, we need to find the {\em top eigenvector} $\ket{\zeta}$ of $M$, i.e., the eigenvector corresponding to its largest eigenvalue. According to the Williamson decomposition~\eqref{williamson}, the Gaussian operator $M$ can be diagonalized by a Gaussian unitary $U$ into a product of thermal states $M_{\text{th}}$,
\begin{equation}\label{meq}
   U^\dagger M U=M_{\text{th}}.
\end{equation}
Since each thermal state is diagonal in the Fock basis, with largest eigenvalue corresponding to the vacuum, the top eigenvector of $M$ is simply  $\ket\zeta=U\ket{0}$, with $\ket{0}$ the $N$-mode vacuum. Denoting by $S$ the symplectic transformation associated with the unitary $U$ in~\eqref{meq}, the state $\zeta\coloneqq|\zeta\rangle\langle\zeta|$ will have covariance matrix $V_\zeta=SS^T$, which can be written in terms of  $V_M$~\eqref{vm}, and hence $V_\rho$ and $V_\sigma$, as a matrix geometric mean~\cite{BHATIA,Gaussian_entanglement_revisited} 
\begin{equation}\label{eq:SST}
    V_\zeta=SS^T=V_M^{\frac12}\Big(V_M^{-\frac12}\Omega V_M^{-1}\Omega^T V_M^{-\frac12}\Big)^{\frac12}V_M^{\frac12}.
\end{equation}
%This is the covariance matrix of the pure Gaussian state $\zeta$ corresponding to the top eigenvector of $M$.

Now, due to definition~\eqref{dmaxmeas}, projecting the operator $M$ onto a given Gaussian state $\ket{\phi}\!\bra{\phi}$ can be done by a Gaussian measurement with seed $\ket{\psi}\!\bra{\psi}$ on the states $\rho$ and $\sigma$. Here, $\ket{\phi}\!\bra{\phi}$ and $\ket{\psi}\!\bra{\psi}$ are related for some fixed $\sigma$ via the map $f_\sigma$,
\begin{equation}\label{eq:f} V_{\phi}=f_\sigma(V_{\psi})\coloneqq V_\sigma-G(V_\sigma \Omega)V_\sigma (V_\sigma +V_{\psi})^{-1}V_\sigma G(\Omega V_\sigma ).
\end{equation}
The domain of $f_\sigma$ is the space of legitimate quantum covariance matrices fulfilling~\eqref{eq:bonafide}.
%, since $\ket{\psi}\!\bra{\psi}$ must correspond to a valid measurement.
In particular, an optimal measurement to achieve $D_{\max}(\rho\|\sigma)$ amounts to projecting $\rho$ and $\sigma$ onto the Gaussian state $\sigma^{-\frac12}\ket{\zeta}$, which has covariance matrix~\cite{lemmings1}
\begin{equation} \label{covmat}
\gamma_{\mathrm{opt}}=-V_\sigma-G(V_\sigma \Omega)V_\sigma(V_\zeta-V_\sigma)^{-1}V_\sigma G(\Omega V_\sigma).
\end{equation}
A Gaussian measurement with seed~\eqref{covmat} always attains the max-relative entropy by construction, but it represents a valid measurement only when $\gamma_{\mathrm{opt}}$ is a  bona fide quantum covariance matrix. When this holds is determined by the following.

\begin{prop}\label{valid}
    The matrix $\gamma_{\mathrm{opt}}$ in~\eqref{covmat} is a legitimate quantum covariance matrix if and only if $\rho$ and $\sigma$ satisfy the inequality
\begin{equation}\label{eq:generalregion}
    V_\zeta<V_\sigma,
\end{equation}
where $V_\zeta$ is expressed in terms of $V_\rho$ and $V_\sigma$ via~\eqref{vm} and~\eqref{eq:SST}.
\end{prop}
\begin{proof}
Begin by assuming~\eqref{eq:generalregion}.  We then have $V_\zeta\geq -i\Omega \implies V_\zeta-V_\sigma\geq -V_\sigma-i\Omega \implies -(V_\zeta-V_\sigma)^{-1}\geq (V_\sigma+i\Omega)^{-1}$.
Using this, we can verify the condition~\eqref{eq:bonafide}:
    \begin{equation*}
    \begin{aligned}
        \gamma_{\mathrm{opt}}+i\Omega&=-V_\sigma-G(V_\sigma \Omega)V_\sigma(V_\zeta-V_\sigma)^{-1}V_\sigma G(\Omega V_\sigma)+i\Omega\\
        &\geq -V_\sigma+G(V_\sigma \Omega)V_\sigma(V_\sigma +i\Omega)^{-1}V_\sigma G(\Omega V_\sigma)+i\Omega\\
        &=-V_\sigma+(V_\sigma - i\Omega)+i\Omega=0,
    \end{aligned}
\end{equation*}
where the third line follows from a variant of~\cite[Lemma 10]{grenyi2018}.  % Lemma~\ref{collapse}.
For the converse, define $X\coloneqq G(V_\sigma \Omega)V_\sigma$. We then have
    $-V_\sigma-X(V_\zeta-V_\sigma)^{-1}X^T\geq 0
    \implies X(V_\zeta-V_\sigma)^{-1}X^T<0
    \implies V_\zeta-V_\sigma<0$,
where for the second implication we use Sylvester's law of inertia.
\end{proof}

Proposition~\ref{valid} tells us that a Gaussian measurement with seed~\eqref{covmat} fulfils the equality~\eqref{equality?} exactly when $\rho$ and $\sigma$ satisfy~\eqref{eq:generalregion}; outside this region, the operator~\eqref{covmat} does not describe a physical state and hence cannot define a valid Gaussian measurement.
As proven in the \hyperref[EndMatter]{Appendix}~A [Lemma~\ref{lemma1}], Eq.~\eqref{eq:generalregion} is in fact the unique region in which a Gaussian measurement achieves the max-relative entropy: any pure measurement giving an output close to $D_{\max}(\rho\|\sigma)$ must be close to the seed in~\eqref{covmat}, and vice versa.
{Physically, the condition $V_\zeta<V_\sigma$ means that the pure Gaussian state $\zeta$ has smaller fluctuations than $\sigma$ along every phase space direction. In the single-mode case presented later, this becomes a direct constraint on the relative squeezing and temperature of $\rho$ and $\sigma$. Experimentally, Eq.~\eqref{eq:generalregion}  can be tested by reconstructing $V_\rho$ and $V_\sigma$ through homodyne tomography~\cite{RevModPhys.81.299}, computing $V_\zeta$ from Eqs.~\eqref{vm} and~\eqref{eq:SST}, and checking that $V_\sigma-V_\zeta>0$.}

%{ Perhaps the most insightful example of~\eqref{eq:generalregion} can be seen in the single-mode case we will discuss later, in which it provides a condition on the relative squeezing and temperature between $\rho$ and $\sigma$. We note this condition may be verified experimentally by reconstructing the covariance matrices of $\rho$ and $\sigma$ using homodyne~\cite{RevModPhys.81.299}, and then using the expressions in~\eqref{vm} and~\eqref{eq:SST} to arrive at the covariance matrix $V_\zeta$.}

We are left to consider the limiting case of singular (e.g., infinitely squeezed) seed covariance matrices, such as homodyne measurements. To do so, we investigate limits of sequences of Gaussian measurements $\left\{\ket{\psi_n}\!\bra{\psi_n}\right\}_n$, where each $\ket{\psi_n}\!\bra{\psi_n}$ is a valid seed.
For each $\ket{\psi_n}$, we define $\ket{\phi_n} \coloneqq {\sigma^{1/2} \ket{\psi_n}}/{\sqrt{\braket{\psi_n| \sigma |\psi_n}}}$. Lemma~\ref{lemma1} tells us that, since $M$ is Gaussian and trace-class, and its largest eigenvalue is non-degenerate, the condition~\eqref{equality?} is equivalent to
$\left|\braket{0|U^\dagger|\phi_n}\right|^2 \tendsnl{?} 1$, which is in turn equivalent %(since the mapping between bona fide covariance matrices and centered Gaussian states is a homeomorphism) 
to $V_{\phi_n} \tendsnl{?} V_\zeta$~\cite{notehomo}\nocite{CV-learning,G-resource-theories,optimal-estimates}. 
Hence, condition~\eqref{equality?} is finally equivalent to asking
whether the closure $\overline{\mathrm{Im}(f_\sigma)}$ of the image of the map $f_\sigma$ defined in~\eqref{eq:f} contains $V_\zeta$.
%\bb
%V_\zeta \overset{?}{\in} \overline{\mathrm{Im}(f_\sigma)}\, .
%\label{problem}
%\ee
This is addressed in \hyperref[EndMatter]{Appendix}~A [Proposition~\ref{prop:saturation}], where we prove that %, if 
 %$V_\sigma > i\Omega$ is any covariance matrix with all symplectic eigenvalues strictly larger than $1$, then 
  $f_\sigma$ satisfies  $V_\zeta \in \overline{\mathrm{Im}(f_\sigma)}$ if and only if
\bb
V_\zeta \leq V_\sigma\, .
\label{saturation_condition_f_sigma}
\ee

Summing up, we have completely solved the main problem~\eqref{equality?}  by establishing the following Theorem.

\begin{thm}\label{maintheorem}
Let $\rho$ and $\sigma$ be two arbitrary $N$-mode Gaussian states with respective covariance matrices $V_\rho$ and $V_\sigma$, with $V_\sigma>V_\rho$.
The achievability of the max-relative entropy $D_{\max}(\rho\|\sigma)$ via Gaussian measurements splits into three cases:
\begin{enumerate}
    \item When $V_\zeta\in \mathrm{Im}(f_\sigma)$, a finitely squeezed Gaussian measurement with seed given by~\eqref{covmat} achieves $D_{\max}(\rho\|\sigma)$. This corresponds to $\rho$ and $\sigma$ satisfying~\eqref{eq:generalregion}.

\item When $V_\zeta\in \overline{\mathrm{Im}(f_\sigma)}$ but $V_\zeta\notin \mathrm{Im}(f_\sigma)$, an infinitely squeezed Gaussian measurement achieves $D_{\max}(\rho\|\sigma)$.  This corresponds to $\rho$ and $\sigma$ on the boundary of~\eqref{saturation_condition_f_sigma}.
    
\item When $V_\zeta\notin\overline{\mathrm{Im}(f_\sigma)}$, there is a strict gap, $D^{\cal G}_{\max}(\rho\|\sigma)<D_{\max}(\rho\|\sigma)$. This corresponds to $\rho$ and $\sigma$ violating~\eqref{saturation_condition_f_sigma}.
\end{enumerate}
\end{thm}

The above fully characterizes when Gaussian measurements are and are not optimal in discriminating between two Gaussian states according to the max-relative entropy. It shows there exists a region of pairs of states $\rho$ and $\sigma$, characterized by Eq.~\eqref{eq:generalregion}, for which $D_{\max}(\rho\|\sigma)$ can be achieved by an optimal Gaussian measurement. At the edge of this region, this optimal measurement becomes infinitely squeezed, corresponding to a homodyne-type detection described by a singular seed. The optimal Gaussian seed becomes unphysical beyond the region of Eq.~\eqref{saturation_condition_f_sigma}, meaning Gaussian measurements are suboptimal in this case, and one must resort to non-Gaussian ones to determine the max-relative entropy between such states. 
{We remark that the achievability condition~\eqref{saturation_condition_f_sigma} is defined only in terms of the covariance matrices of $\rho$ and $\sigma$ and is independent of their displacements, as shown in detail in \hyperref[EndMatter]{Appendix}~B.}

Case 3 in Theorem~\ref{maintheorem} establishes an original {\em``no-go''} result for `all-Gaussian' quantum state discrimination, adding to the existing limitations of Gaussian states and operations in quantum information~\cite{G-resource-theories,fundamista,noGdistillation,noGdistillation2,noGdistillation3,nogoerror}. Perhaps more surprisingly, the {\em ``yes-go''} Cases 1 and 2 reveal how simple Gaussian measurements are in fact optimally suited to discriminate between all pairs of multimode states with second moments fulfilling~\eqref{eq:generalregion} according to $D_{\max}$. This is in  contrast with other distinguishability metrics such as fidelity, for which Gaussian measurements are only optimal in singular cases~\cite{Oh_2019no2}.

\paragraph{\bfseries Single-mode case.}
We can illustrate our general results focusing on the single-mode case. This is extended to tensor products of single-mode states in \hyperref[EndMatter]{Appendix}~C  [Proposition~\ref{prop:separates}].
Let $\rho$ and $\sigma$ be two undisplaced single-mode Gaussian states. Via symplectic operations that leave $D_{\max}(\rho\|\sigma)$ invariant, we can transform $\rho$ and $\sigma$ into a squeezed thermal and a thermal state, respectively, with standard form covariance matrices 
\begin{equation}\label{VW1}
V_\rho =  a \ {\rm diag}(\mu, 1/\mu) \mbox{ \quad and \quad } V_\sigma =  b \ \id_2, 
\end{equation}
with $a \geq 1$, $b>1$, and $\mu\in (0,\infty)$. The covariance matrix $V_M$~\eqref{vm} can be written as $V_M = {\rm diag} \left(\frac{ab\mu-1}{b-a\mu}, \frac{\mu-ab}{a-b\mu}\right)$.
The Gaussian unitary $U$ that transforms $M$ into a thermal state as in~\eqref{meq} is the squeezing operation that equalizes the magnitudes of the eigenvalues of the Gibbs matrix $\Gamma_M\coloneqq2i\Omega \ \text{arccoth} (V_M i\Omega)$~\cite{fidelity}. This corresponds to the symplectic operation $S = {\rm diag}\left(\left(\frac{(a-b\mu)(ab\mu-1)}{(ab-\mu)(a\mu-b)}\right)^{\frac14}, \left(\frac{(ab-\mu)(a\mu-b)}{(a-b\mu)(ab\mu-1)}\right)^{-\frac14}\right)$.

We start with the case in which a Gaussian measurement can achieve the max-relative entropy, i.e.\ Case 1 in Theorem~\ref{maintheorem}. 
For $V_\rho$ and $V_\sigma$ in the given standard form, the optimal measurement has a pure seed covariance matrix of the form
    $\gamma(z)={\rm diag}(z, 1/z)$,
corresponding to a generaldyne projection~\cite{BUCCO} parametrized by the squeezing  $z\in \mathbb{R}_+$. 
We have then
\begin{equation}\label{DGmax1}
D^{\cal G}_{\max}(\rho\|\sigma) = \sup_{z\in \mathbb{R}_+} \frac12 \ln \left(\frac{\mu  (b+z) (b z+1)}{(a z+\mu ) (a \mu +z)}\right).
\end{equation}
The squeezing value $z_{\mathrm{opt}}$ realizing the optimal measurement seed $\gamma_{\mathrm{opt}}\equiv \gamma(z_{\mathrm{opt}})$ can be found from~\eqref{covmat}, yielding
$z_{\mathrm{opt}} =\frac{a b \left(\mu ^2-1\right)+\sqrt{(a b-\mu ) (a \mu -b) (a-b \mu ) (a b \mu -1)}}{a \left(b^2+1\right)-\left(a^2+1\right) b \mu }$.
%This solution achieves $D^{\cal G}_{\max}(\rho\|\sigma)=D_{\max}(\rho\|\sigma)$.  
This corresponds to a valid Gaussian measurement --- i.e.\ $z_{\mathrm{opt}}\in (0,\infty)$ --- in the region 
\begin{equation} \label{interval}
    \frac{b(a^2+1)}{a(b^2+1)}\eqqcolon\mu_{\min}<\mu<\mu_{\max} \coloneqq\frac{a(b^2+1)}{b(a^2+1)},
\end{equation}
that is the region~\eqref{eq:generalregion} in which the optimal Gaussian measurement seed is physical, according to Proposition~\ref{valid}.

Now we move to Cases 2 and 3 of Theorem~\ref{maintheorem}. For $V_\rho,V_\sigma$ as in~\eqref{VW1} with $\mu$ outside of~\eqref{interval}, the optimal Gaussian measurement reduces to a homodyne detection of either the position or the momentum quadrature, corresponding to the limit $z\to 0$ or $z\to \infty$ in $\gamma(z)$.
At the boundaries of the region~\eqref{interval},  the max-relative entropy is still reached  in the sense that we can find a sequence of Gaussian measurements that gives outcomes approximating that of the ideal homodyne detection (Case $2$ in Theorem~\ref{maintheorem}). For states that do not satisfy $\mu_{\min} \leq \mu \leq \mu_{\max}$,
there is a strict gap between $D^{\cal G}_{\max}(\rho\|\sigma)$ and $D_{\max}(\rho\|\sigma)$, corresponding to Case 3 in Theorem~\ref{maintheorem}.

The Gaussian measured max-relative entropy is, explicitly,
\begin{equation}\label{DGmax1final}
\!\!\!\!D^{\cal G}_{\max}(\rho\|\sigma) =  \left\{  \!\!
\begin{array}{cc}
% \frac{1}{2} \ln \left(\frac{\left(a^2+1\right) \left(b^2+1\right) \mu -2 a b (\mu ^2+1)-2 \sqrt{(a b-\mu ) (a \mu -b) (a-b \mu ) (a b \mu -1)}}{\left(a^2-1\right)^2 \mu }\right), 
\frac12 \ln \left(\frac{\mu  (b+z_{\mathrm{opt}}) (b z_{\mathrm{opt}}+1)}{(a z_{\mathrm{opt}}+\mu ) (a \mu +z_{\mathrm{opt}})}\right),
\qquad &\frac{b(a^2 +1)}{a(b^2+1)}\!<\mu <\!\frac{a(b^2+1)}{b(a^2+1)} \\
 \frac{1}{2} \ln \left(\frac{b \mu }{a}\right), & \frac{a(b^2+1)}{b(a^2+1)} \leq \mu < \frac{b}{a} \\
 \frac{1}{2} \ln \left(\frac{b}{a \mu }\right), & \frac{a}{b} < \mu \leq \frac{b(a^2 +1)}{a(b^2+1)} \\
  \infty,  & b \mu <a\lor a \mu > b.  \\
\end{array}\right.\!\!\!\!\!\!\!
\end{equation}

%\begin{figure}[t]
%\centering
%\includegraphics[width=4.2cm]{dgn05labeled.png}\hspace*{0.1cm}
%\includegraphics[width=4.2cm]{dgn20labeled.png} \\[-0.2cm]
\begin{figure}[t]
\centering
{
\sffamily
\sansmath
\begin{tikzpicture}
\node[anchor=south west,inner sep=0] (img1) at (0,0)
{\includegraphics[width=4.15cm]{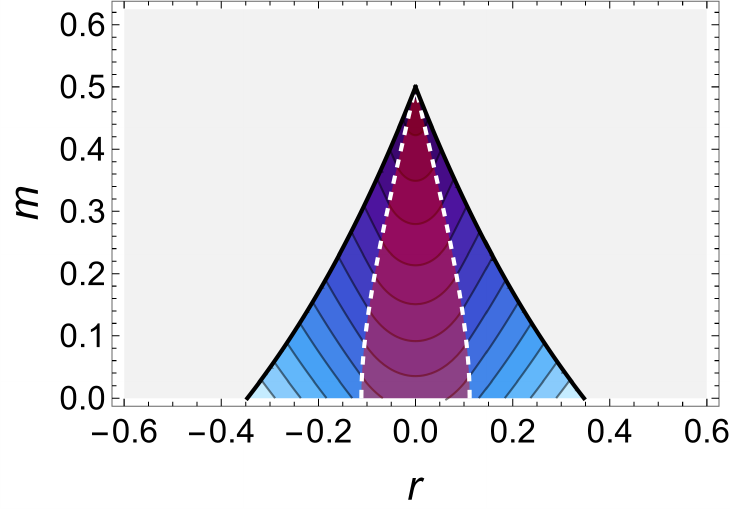}};
\begin{scope}[x={(img1.south east)},y={(img1.north west)}]
    \node[overlay,font=\fontfamily{ua1}\selectfont\scriptsize,white,rotate=80] at (0.46,0.35) {$\mu_{\min}$};
    \node[overlay,font=\fontfamily{ua1}\selectfont\scriptsize,white,rotate=-80] at (0.67,0.35) {$\mu_{\max}$};
    \node[overlay,anchor=north east,font=\fontfamily{ua1}\selectfont\scriptsize] at (0.95,0.95) {$n=0.5$};
\end{scope}
\end{tikzpicture}
\hspace{0.08cm}
\begin{tikzpicture}
\node[anchor=south west,inner sep=0] (img2) at (0,0)
{\includegraphics[width=4.15cm]{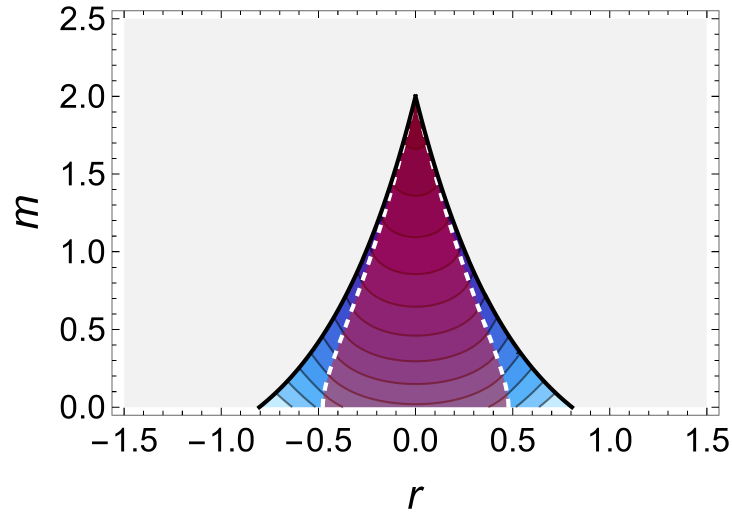}};
\begin{scope}[x={(img2.south east)},y={(img2.north west)}]
    \node[overlay,font=\fontfamily{ua1}\selectfont\scriptsize,white,rotate=72] at (0.5,0.35) {$\mu_{\min}$};
    \node[overlay,font=\fontfamily{ua1}\selectfont\scriptsize,white,rotate=-68] at (0.633,0.35) {$\mu_{\max}$};
    \node[overlay,anchor=north east,font=\fontfamily{ua1}\selectfont\scriptsize] at (0.95,0.95) {$n=2.0$};
\end{scope}
\end{tikzpicture}
}
\vspace*{-0.6cm}
\caption{\label{OKeeffe}
{ Contour plot of the Gaussian measured max-relative entropy $D^{\cal G}_{\max}(\rho\|\sigma)$ (shades of blue, increasing from darker to lighter) between a squeezed thermal state $\rho$ with covariance matrix $V_\rho=(2m+1) {\rm diag}({\rm e}^{2r}, {\rm e}^{-2r})$ and a thermal state $\sigma$ with covariance matrix $V_\sigma=(2n+1) \id_2$, for $n=1/2$ (left) and $n=2$ (right).  In the gray background region,  $D^{\cal G}_{\max}(\rho\|\sigma)=\infty$. In the reddish shaded central region, $D^{\cal G}_{\max}(\rho\|\sigma)=D_{\max}(\rho\|\sigma)$, that is, Gaussian measurements are optimal. The white dashed boundary spans ${\rm e}^{2r}=\mu_{\text{min}},\mu_{\text{max}}$~\eqref{interval}, where a homodyne measurement achieves $D_{\text{max}}$ as a limiting case.}} 
\end{figure}

\paragraph{\bfseries Gaussian data hiding.}
One can wonder how large the difference between $D_{\max}(\rho\|\sigma)$ and $D^{\cal G}_{\max}(\rho\|\sigma)$ can be when Gaussian measurements are {\it not} optimal. We can contrast the expression in~\eqref{DGmax1final} with the formula~\eqref{gdmaxformula} for the unrestricted max-relative entropy $D_{\max}(\rho\|\sigma)$, between the same pair of states~\eqref{VW1}.
%,~\eqref{gdmaxformula}, which in this single-mode case gives
%\begin{equation}\label{Dmax1final}
%D_{\max}(\rho\|\sigma) =  \left\{  
%\begin{array}{cc}
% \frac{1}{2} \ln \left(\frac{\left(a^2+1\right) \left(b^2+1\right) \mu -2 a b (\mu ^2+1)-2 \sqrt{(a b-\mu ) (a \mu -b) (a-b \mu ) (a b \mu -1)}}{\left(a^2-%1\right)^2 \mu }\right), \qquad \qquad & \frac{a}{b} < \mu < \frac{b}{a} \\
%  \infty,  & b \mu <a\lor a \mu > b.  \\
%\end{array}\right.
%\end{equation}
Let us adopt a physically natural parameterization, by letting $a=2m+1$, $b=2n+1$, and $\mu={\rm e}^{2r}$, in which $m,n$ are mean thermal photon numbers for $V_\rho$ and $V_\sigma$, respectively, and $r\in \mathbb{R}$ is the squeezing factor for $V_\rho$. The resulting analysis is illustrated in Fig.~\ref{OKeeffe}, highlighting the parameter regions corresponding to the different cases in Theorem~\ref{maintheorem}.

Qualitatively, Gaussian measurements are optimal when $\rho$ is not excessively squeezed, and the region of Gaussian optimality broadens with increasing temperature of $\sigma$. But how inferior can Gaussian measurements be outside of this region?

Take $\rho$ and $\sigma$ to be nearly pure and with the first state's squeezing being as large as possible within the allowed region of finite max-relative entropy (i.e.\ a state in the bottom right or bottom left corner of the blue shaded regions in Fig.~\ref{OKeeffe}). Specifically, let $a=2m+1=1+\epsilon$, $b=2n+1=1+\kappa \epsilon$, and $\mu = {\rm e}^{2r} = 1+ (\kappa-1) \epsilon (1-\epsilon )$, with $0 <\epsilon <1$ and $\kappa>1$. For these pairs of states, by taking series expansions in $\epsilon$, we get
\begin{equation}\label{datahidden}
\begin{aligned}
D^{\cal G}_{\max}(\rho\|\sigma) &\approx (\kappa-1)\epsilon + O(\epsilon^2)\,, \\
D_{\max}(\rho\|\sigma) &\gtrapprox \frac{\ln\kappa}{2}-\frac{\kappa \epsilon }{4}+ O(\epsilon^2)\,.
\end{aligned}
\end{equation}
This shows that, as long as $\kappa \gg 1$ and $\epsilon \ll \kappa^{-1}$, the Gaussian measured max-relative entropy is vanishingly small while the unrestricted max-relative entropy, which can be attained by suitable non-Gaussian measurements, can be arbitrarily large.

The above may be interpreted as a form of \emph{Gaussian data hiding}~\cite{PhysRevA.104.052428,sabapathy2021bosonicdatahidingpower,wang2025gaussianquantumdatahiding}, in which the states we find are almost indistinguishable via Gaussian measurements, even though if we allow all measurements the states are very distinguishable through their max-relative entropy. % This demonstrates a remarkable limitation of Gaussian measurements, even when working with only Gaussian states. 
%Interestingly, the gap in the discriminating power of Gaussian measured versus general unmeasured quantities is not specific to $D_{\max}$, but persists for the extended family of sandwiched R\'enyi divergences, see \hyperref[EndMatter]{Appendix}~C  [Remark~\ref{remalpha}].
{This  phenomenon has a direct operational interpretation in maximum-confidence discrimination~\cite{Croke_2006,Mosley_2006}. Recall from Eq.~\eqref{eq:maxconfidence} that $C_\rho^\star(p)$ is the largest posterior probability with which a conclusive measurement outcome can identify $\rho$, given prior probability $p$. Restricting the measurement to be Gaussian defines analogously $C_{\rho}^{\star\G}(p)=p\,{\rm e}^{D_{\max}^{\G}(\rho\|\sigma)}\big/\,\big[1-p+p\,{\rm e}^{D_{\max}^{\G}(\rho\|\sigma)}\big]$; see \hyperref[app:maxconfidence]{Appendix}~D for details. For equal priors ($p=1/2$), choosing for instance $\epsilon=\kappa^{-2}$ and taking $\kappa\to\infty$ yields $D_{\max}^{\G}\to0$ while $D_{\max}\to\infty$, and hence
$C_{\rho}^{\star\G}\longrightarrow \frac12$ while $C_\rho^\star\longrightarrow 1$.
Thus, Gaussian receivers provide asymptotically no confidence beyond the prior, whereas unrestricted measurements can identify the Gaussian state $\rho$ with arbitrarily high confidence.}

\paragraph{\bfseries Conclusion.}
%\section{Conclusion}\label{sec:conclusion}
In this Letter, we have characterized the distinguishability of bosonic Gaussian states in quantum information, and specifically addressed the question of whether Gaussian measurements are optimal for determining the max-relative entropy between two Gaussian states. We found there exists a region of states, defined by an inequality on their covariance matrices, for which there is a Gaussian measurement achieving the max-relative entropy. At the boundary of this region, the optimal Gaussian measurement becomes a homodyne-type detection. In turn, outside this region, Gaussian measurements are strictly suboptimal and non-Gaussian measurements are required to achieve the max-relative entropy. 

As an insightful case study, we derived explicit formulas for the Gaussian measured max-relative entropy between two single-mode Gaussian states in terms of their mean photon numbers and relative squeezing parameter, and presented examples of states for which the gap between the max-relative entropy and the Gaussian measured max-relative entropy can be made arbitrarily large, exhibiting a data hiding phenomenon.

Our results hold in the general multimode case and clarify the capabilities and limitations of Gaussian measurements in quantum state discrimination. 
{Operationally, they provide an exact criterion for when a Gaussian receiver can attain  maximum confidence in distinguishing between  two Gaussian hypotheses.}
%In the data hiding regime, the limitations become extreme: unrestricted measurements can identify one hypothesis with confidence approaching unity, while Gaussian measurements provide essentially no improvement over the prior probability.
Future work may explore other {information-theoretic protocols}  under Gaussian measurement restrictions, as well as the discrimination of non-Gaussian states.

\smallskip

\begin{acknowledgments}
We acknowledge financial support from the Engineering and Physical Sciences Research Council (EPSRC Grants No.~EP/W524402/1, EP/X010929/1, and EP/T022140/1) and the European Union under the European Research Council (ERC Grant Agreement No.~101165230).
\end{acknowledgments}

%\clearpage
%\onecolumngrid
\appendix
\section{APPENDIX: END MATTER}\label{EndMatter}

%Here we present technical derivations and extensions of the main results. 

\subsection{A.~Proof of the achievability condition~\eqref{saturation_condition_f_sigma}}

For two positive semi-definite operators $A,B\geq 0$ and some $t\in \R$, set the matrix geometric mean~\cite{BHATIA}
\bb
A\#_t B \coloneqq A^{1/2} \big(A^{-1/2} B A^{-1/2}\big)^t A^{1/2} .
\ee
Recall that $A\#_t B > 0$ if both $A,B>0$, and, although not obvious from the above formula, $A\#_t B = B \#_{1-t} A$.
Then we can rewrite~\eqref{eq:SST} as~\cite[Lemma~12]{Gaussian_entanglement_revisited}
\begin{equation}\label{eq:SST2}
    V_\zeta=SS^T=V_M \# \big(\Omega V_M^{-1}\Omega^T\big).
\end{equation}

We will need the following lemmas.

\begin{lemma} \label{lemma1}
     Let $A,B>0$ be bounded operators such that $A\leq B$ and define $T:=B^{\sfrac{-1}{2}}AB^{\sfrac{-1}{2}}$. Let the largest eigenvalue of $T$, $\lambda_1(T)$, have associated eigenvector $\ket{0}$, and let the second largest eigenvalue be such that there is a strict gap, i.e. $\lambda_1(T)>\lambda_2(T)$. If $\bra{\psi}A\ket{\psi}\geq\bra{\psi}B\ket{\psi}(\lambda_1-\varepsilon)$, then the vector $\ket{\phi}\coloneqq \frac{B^{\sfrac{1}{2}}\ket{\psi}}{\sqrt{\bra{\psi}B\ket{\psi}}}$ satisfies $|\braket{\phi|0}|^2\geq 1-\frac{\varepsilon}{\lambda_1-\lambda_2}$.
     Vice versa, if some normalized vector $\ket{\psi}$ is such that the corresponding $\ket{\phi}$ satisfies $|\braket{\phi|0}|^2\geq 1-\frac{\delta}{\lambda_1}$, then $\bra{\psi}A\ket{\psi}\geq\bra{\psi}B\ket{\psi}(\lambda_1-\delta)$.
\end{lemma} 

\begin{proof}
We start from the first claim. We have $\bra{\phi}T\ket{\phi}=\frac{\bra{\psi}B^{\sfrac{1}{2}}B^{\sfrac{-1}{2}}AB^{\sfrac{-1}{2}}B^{\sfrac{1}{2}}\ket{\psi}}{\bra{\psi}B\ket{\psi}}=\frac{\bra{\psi}A\ket{\psi}}{\bra{\psi}B\ket{\psi}}$. Expanding $\ket{\phi}=\sqrt{p}\ket{0}+\sqrt{1-p}\ket{1}$, where $\ket{1}$ is orthogonal to $\ket{0}$ and $p:=|\braket{\phi|0}|^2$, yields
       $\bra{\phi}T\ket{\phi}=p\bra{0}T\ket{0}+(1-p)\bra{1}T\ket{1}
        \leq p\bra{0}T\ket{0}+(1-p)\lambda_2
        =p\lambda_1+(1-p)\lambda_2$.
Since $\bra{\phi}T\ket{\phi}=\lambda_1-\varepsilon$ by assumption, we arrive at $|\braket{\phi|0}|^2\geq 1-\frac{\varepsilon}{\lambda_1-\lambda_2}$.

 As for the second claim, assuming that $|\braket{\phi|0}|^2\geq 1-\frac{\delta}{\lambda_1}$, we write
$\frac{\bra{\psi}A\ket{\psi}}{\bra{\psi}B\ket{\psi}} = \bra{\phi}B^{-1/2}AB^{-1/2}\ket{\phi} \geq \bra{\phi} \left( \lambda_1 \ketbra{0} \right)\ket{\phi} = \lambda_1 |\braket{\phi|0}|^2 \geq \lambda_1 - \delta$,
concluding the proof.
\end{proof}

\begin{lemma} \label{strict_monotonicity_geometric_mean_lemma}
Let $A,B,C>0$ be three positive definite matrices, with $A < B$. Then, for all $t\in [0,1)$, we have $A \#_t C < B \#_t C$.
\end{lemma}

\begin{proof}
The condition $0<A<B$ means %$\exists \ 
that there exists $\lambda>1$ such that $\lambda A \leq B$. Taking the geometric mean with $C$ of both sides gives
$A \#_t C < \lambda^{1-t} A \#_t C = (\lambda A) \#_t C \leq B \#_t C$,
where the first inequality holds because $A \#_t C > 0$, and the second by the monotonicity of the geometric mean in each of its arguments.
\end{proof}

\begin{lemma} \label{geometric_mean_pure_QCMs_lemma}
Let $\gamma_1,\gamma_2$ be the quantum covariance matrices of two pure states. Then, for any $t\in \R$, their geometric mean $\gamma_1 \#_t \gamma_2$ is also the quantum covariance matrix of a pure state.
\end{lemma}

\begin{proof}
%A reasoning entirely analogous to that used to prove~\cite{LL-log-det} shows that a positive definite matrix $V$ is a bona fide quantum covariance matrix, i.e.\ it satisfies $V\geq i\Omega$, if and only if $V\geq \Omega V^{-1} \Omega^T$. Writing this out for $\gamma_i$ ($i=1,2$), we get $\gamma_i \geq \Omega \gamma_i^{-1} \Omega^T$
A simple calculation using properties of symplectic matrices and Williamson's theorem shows that a positive definite matrix $V$ is the quantum covariance matrix of a pure state if and only if $V=\Omega V^{-1} \Omega^T$ (see e.g.~\cite[Lemma~11]{LL-log-det}). Writing this for $\gamma_i$ ($i=1,2$) we get $\gamma_i = \Omega \gamma_i^{-1} \Omega^T$. Taking the geometric mean of both sides gives
$\gamma_1 \#_t \gamma_2 = \big(\Omega \gamma_1^{-1} \Omega^T\big) \#_t \big(\Omega \gamma_2^{-1} \Omega^T\big) = \Omega \big(\gamma_1^{-1} \#_t \gamma_2^{-1} \big) \Omega^T = \Omega (\gamma_1 \#_t \gamma_2)^{-1} \Omega^T$.
Hence, $\gamma_1 \#_t \gamma_2$ is the quantum covariance matrix of a pure state.
\end{proof}

By writing $\gamma^t = \id \#_t \gamma$, we get the following corollary.

\begin{cor} \label{powers_cor}
Let $\gamma$ be the quantum covariance matrix of a pure state. Then, for any $t\in \R$, the matrix power $\gamma^t$ is also the quantum covariance matrix of a pure state.
\end{cor}

We are now ready to solve the main problem~\eqref{equality?}, that is equivalent to $V_\zeta \overset{?}{\in} \overline{\mathrm{Im}(f_\sigma)}$, with $V_\zeta = SS^T$ given in~\eqref{eq:SST2}. 

\begin{prop}\label{prop:saturation}
Let $V_\sigma > i\Omega$ be any quantum covariance matrix with all symplectic eigenvalues strictly larger than $1$, and let $S$ be a symplectic matrix. Then the map $f_\sigma$ defined by~\eqref{eq:f} satisfies that $SS^T \in \overline{\mathrm{Im}(f_\sigma)}$ if and only if
\bb
SS^T \leq V_\sigma\, .
\label{saturation_condition_f_sigma2}
\ee
\end{prop}

\begin{proof}
That~\eqref{saturation_condition_f_sigma2} is necessary is clear from the fact that any matrix $f_\sigma(W) \in \mathrm{Im} (f_\sigma)$ satisfies $f_\sigma(W) < V_\sigma$. By taking the closure, we see that any matrix in $\overline{\mathrm{Im}(f_\sigma)}$ must also satisfy the same (non-strict) inequality.

We will now prove that~\eqref{saturation_condition_f_sigma2} is also sufficient. Let $V_\sigma = Z D Z^T$ be the Williamson decomposition of the quantum covariance matrix $V_\sigma$. By assumption, $D > \id$, so that the pure quantum covariance matrix $\gamma \coloneqq ZZ^T$ satisfies the inequality $\gamma < V_\sigma$~\cite[Lemma~11]{LL-log-det}. Now, for any $t\in \R$ set
\bb
S_t \coloneqq \big( (SS^T) \#_t \gamma\big)^{1/2} (SS^T)^{-1/2} S\, .
\ee
%where $A\#_t B$ is the matrix geometric mean between positive semi-definite matrices. 
We have:
\begin{enumerate}[(a)]

\item $S_t$ is a symplectic matrix for all $t\in \R$. In fact: $SS^T$ is the quantum covariance matrix of a pure state; hence, the same is true of $(SS^T) \#_t \gamma$, due to Lemma~\ref{geometric_mean_pure_QCMs_lemma}; by Corollary~\ref{powers_cor}, also $\big( (SS^T) \#_t \gamma\big)^{1/2}$ corresponds to a pure state, and is therefore symplectic; since the same can be said of $(SS^T)^{-1/2}$, and the product of symplectic matrices is symplectic, the claim is established.

\item $S_t^{\vphantom{T}} S_t^T < V_\sigma$ for all $t\in (0,1]$.  One finds that
$S_t^{\vphantom{T}} S_t^T = \big( (SS^T) \#_t \gamma\big)^{1/2} (SS^T)^{-1/2} S S^T (SS^T)^{-1/2} \big( (SS^T) \#_t \gamma\big)^{1/2} = (SS^T) \#_t \gamma < V_\sigma$,
where the last inequality holds due to Lemma~\ref{strict_monotonicity_geometric_mean_lemma}, because $\gamma < V_\sigma$.

\item $\lim_{t\to 0} S_t = S$: in fact, $t\mapsto S_t$ is manifestly continuous, and $S_0 = S$ by construction.
\end{enumerate}

From~(a) and~(b), we have $S_t^{\vphantom{T}} S_t^T \in \mathrm{Im}(f_\sigma)$ for all $t\in (0,1]$. Thus, (c)~guarantees $SS^T \in \overline{\mathrm{Im}(f_\sigma)}$, concluding the proof.
\end{proof}

\subsection{B.~Remarks on arbitrary displacements}\label{app:nodisplacements}
In the previous section, we solved the main problem~\eqref{equality?} under the assumption of undisplaced states. As anticipated ~\cite{notedisplacement}, our resulting characterization of the achievability of the max-relative entropy by Gaussian measurements holds in full generality, regardless of the displacements of the involved states. 
Given two $N$-mode Gaussian states $\rho$, and $\sigma$ such that $V_\rho < V_\sigma$, we are asking ourselves under what conditions
$$D_{\max}^\G(\rho\|\sigma) \coloneqq \sup_{\MM_\G} D_{\max}\left(\MM_\G(\rho)\|\MM_\G(\sigma)\right) \eqt{?} D_{\max}(\rho\|\sigma),$$
where the supremum is over all Gaussian measurements $\MM_\G$. This is the same as asking whether there exists a sequence of Gaussian measurements with POVM elements $\left( D(s) \psi_n D(-s) \right)_{s\in \R^{2N}}$, with $D(s)$ the displacement operator and $\psi_n = \ketbra{\psi_n}$ being a zero-mean, pure Gaussian state (the seed), such that
\begin{equation}
\lim_{n\to\infty} \sup_{s\in \R^{2N}} \frac{\braket{\psi_n|D(-s) \rho D(s) |\psi_n}}{\braket{\psi_n|D(-s) \sigma D(s) |\psi_n}} \eqt{?} \big\|\sigma^{-1/2} \rho \sigma^{-1/2}\big\|_\infty\, .
\label{equality_1}
\end{equation}

Define $\ket{\xi_n(s)} \coloneqq \frac{\sigma^{1/2} D(s) \ket{\psi_n}}{\sqrt{\braket{\psi_n|D(-s) \sigma D(s) |\psi_n}}}$.
We observe that, for every $n$, the set of states $\big\{\ket{\xi_n(s)}\big\}_{s\in \R^{2N}}$ coincides with the set of states $\big\{\ket{\phi_n(t)}\coloneqq D(t)\ket{\phi_n}\big\}_{t\in \R^{2N}}$,
where $\ket{\phi_n} \coloneqq \frac{\sigma_0^{1/2} \ket{\psi_n}}{\sqrt{\braket{\psi_n| \sigma_0 |\psi_n}}}$,
with $\sigma_0$ being the Gaussian state with zero mean and the same covariance matrix as $\sigma$. Indeed, the `golden rule' of~\cite[Appendix~A]{petzrecovery} entails that the first moments of $\xi_n(s)$ depend affinely on $s$, with the matrix multiplying $s$ being invertible. The covariance matrix, instead, is the same as that of $\phi_n$. 

By Lemma~\ref{lemma1}, the condition~\eqref{equality_1} is therefore equivalent to
\bb
\sup_{t\in \R^{2N}} \big|\braket{\eta|\phi_n(t)}\big|^2 \tendsnl{?} 1\, ,
\label{overlap}
\ee
with $\ket{\eta}$ the eigenvector of $\sigma^{-1/2} \rho \sigma^{-1/2}$ with maximal eigenvalue. 
Due to the fact that the overlap between two Gaussian states with fixed covariance matrices is maximal when also the means coincide~\cite[Eq.~(4.51)]{BUCCO}, the supremum in~\eqref{overlap} is  achieved when $t$ coincides with the mean of $\eta$. This implies that~\eqref{overlap} is equivalent to requiring that
$\ket{\phi_n} \tendsnl{?} \ket{\eta_0}$,
where $\eta_0$ has the same covariance matrix as $\eta$ but zero mean, and the convergence is in trace norm (or, equivalently, in fidelity).

\subsection{C.~Extension to tensor product states}

Here we extend our single-mode results, including Eq.~\eqref{DGmax1final}, to multimode tensor product states. We show that the optimal Gaussian measurement factorizes into individual optimal single-mode measurements, which can be constructed as in the main text by treating each constituent individually.
\begin{prop}\label{prop:separates}
Let $\rho,\sigma$ be zero-mean Gaussian states such that there exists a symplectic matrix simultaneously block-diagonalizing $V_\rho$ and $V_\sigma$, and let $\rho^{(j)}$ and $\sigma^{(j)}$ be reduced density matrices corresponding to the $j\mathrm{th}$ block. We have then
   \begin{equation}
        D^{\cal G}_{\max}(\rho\|\sigma)={\sum}_{j=1}^N  D^{\cal G}_{\max}\big(\rho^{(j)}\,\big\|\,\sigma^{(j)}\big).
   \end{equation} 
\end{prop}

\begin{proof}
Clearly $ D^{\cal G}_{\max}(\rho\|\sigma)\geq\sum_{j=1}^N  D^{\cal G}_{\max}(\rho^{(j)}\|\sigma^{(j)})$, since an optimal measurement on each mode will be a candidate for the optimal global measurement, and for this choice of $\gamma=\bigoplus_{j=1}^N\gamma^{(j)}$ we have 
$D^{\cal G}_{\max}(\rho\|\sigma) \geq\frac12\ln\frac{\det\left(\bigoplus_{j=1}^N \left(V_\sigma^{(j)}+\gamma^{(j)}\right)\right)}{\det\left(\bigoplus_{j=1}^N \left(V_\rho^{(j)}+\gamma^{(j)}\right)\right)}=\sum_{j=1}^N\frac12 \ln\frac{\det\left(V_\sigma^{(j)}+\gamma^{(j)}\right)}{\det\left(V_\rho^{(j)}+\gamma^{(j)}\right)}=\sum_{j=1}^N  D^{\cal G}_{\max}(\rho^{(j)}\|\sigma^{(j)})$.

For the converse, consider a general measurement seed with covariance matrix  $\gamma$ written in block form $\gamma=(\gamma_{jk})_{j,k=1}^N$.
Because $V_\rho = \bigoplus_{j=1}^N V_\rho^{(j)}$ is block-diagonal, the $j$th diagonal block of the inverse matrix $(V_\rho+\gamma)^{-1}$ is exactly $(V_\rho^{(j)}+\gamma^{(j)})^{-1}$, where $\gamma^{(j)}$ is given by the Schur complement of the remaining blocks, $\gamma^{(j)} \coloneqq \gamma_{jj} - \gamma_{j\bar{\jmath}} \big(V_\rho^{(\bar{\jmath})} + \gamma_{\bar{\jmath}\bar{\jmath}}\big)^{-1} \gamma_{\bar{\jmath}j}$,
with $\bar{\jmath}$ denoting all modes except the $j$th. The covariance matrix $\gamma^{(j)}$ represents a valid measurement seed for the $j$th block.
We can then write
\begin{equation*}
\resizebox{\linewidth}{!}{%
$\begin{aligned}
  & D^{\cal G}_{\max}(\rho\|\sigma)=\frac12\sup_{\gamma\geq i\Omega}\ln\frac{\det\left(\bigoplus_{j=1}^N V_\sigma^{(j)}+\gamma\right)}{\det\left(\bigoplus_{j=1}^N V_\rho^{(j)}+\gamma\right)} \\   
&=\frac12\sup_{\gamma\geq i\Omega}\ln\det\left(\id+\left(\!\bigoplus_{j=1}^N\! V_\sigma^{(j)}\!-\!V_\rho^{(j)}\right)^\frac12\! \left(\!\bigoplus_{j=1}^N\! V_\rho^{(j)}\!+\!\gamma\right)^{-1}\!\left(\!\bigoplus_{j=1}^N\! V_\sigma^{(j)}\!-\!V_\rho^{(j)}\right)^\frac12\right)\\    
&\leq\frac12\ln\left[\prod_{j=1}^N\sup_{\gamma^{(j)}\geq i\Omega^{(j)}}\det\left(\id+\left(V_\sigma^{(j)}-V_\rho^{(j)}\right)^\frac12 \left(V_\rho^{(j)}+\gamma^{(j)}\right)^{-1}\left(V_\sigma^{(j)}-V_\rho^{(j)}\right)^\frac12\right)\right]\\   
&=\sum_{j=1}^N\frac12\sup_{\gamma^{(j)}\geq i\Omega^{(j)}}\ln\frac{\det\left(V_\sigma^{(j)}+\gamma^{(j)}\right)}{\det\left(V_\rho^{(j)}+\gamma^{(j)}\right)}=\sum_{j=1}^N
    D^{\cal G}_{\max}(\rho^{(j)}\|\sigma^{(j)}),
    \end{aligned}$}
\end{equation*}
where to get the inequality in the third line we use Fischer's inequality ($\det A\leq \prod_k \det A_{kk}$ for a block matrix $A \geq 0$).
\end{proof}

{\subsection{D.~Maximum-confidence operational interpretation of $\mathbf{{\emph{D}}_{\mathbf{max}}}$}}
\label{app:maxconfidence}
{Here we derive the maximum-confidence interpretation of the max-relative entropy used in the text. Consider two hypotheses described by quantum states $\rho$ and $\sigma$, prepared with prior probabilities $p$ and $1-p$, respectively. Let $E$ be a measurement event interpreted as identifying $\rho$. Conditioned on this outcome, the posterior probability that the state was $\rho$ is
\begin{equation}
C_\rho(E)=
\frac{p\,\tr[\rho E]}
{p\,\tr[\rho E]+(1-p)\,\tr[\sigma E]} .
\end{equation}
The maximum confidence $C_\rho^\star(p)$ is the largest such posterior probability over all measurement events $E$. Writing $R_E\coloneqq\tr[\rho E]/\tr[\sigma E]$, we have $C_\rho(E)=pR_E/(1-p+pR_E)$, which increases monotonically with $R_E$. Eq.~\eqref{dmaxmeas} then gives Eq.~\eqref{eq:maxconfidence}, namely $C_\rho^\star(p)=p\,{\rm e}^{D_{\max}(\rho\|\sigma)}/[1-p+p\,{\rm e}^{D_{\max}(\rho\|\sigma)}]$, recovering the maximum-confidence discrimination problem of Ref.~\cite{Croke_2006}.

The same reasoning applies when the physical measurement is restricted to be Gaussian. For a Gaussian measurement $\MM_\G$ with continuous outcome $x$, let $P_\rho^{\MM_\G}$ and $P_\sigma^{\MM_\G}$ denote the corresponding outcome distributions. Their classical max-relative entropy satisfies
$\exp\big[D_{\max}\big(P_\rho^{\MM_\G}\,\big\|\,P_\sigma^{\MM_\G}\big)\big]
=
\sup\limits_{x}
\frac{P_\rho^{\MM_\G}(x)}{P_\sigma^{\MM_\G}(x)}$.
Equivalently, this supremum can be approached by postselecting suitable measurable regions of the outcomes. Defining $C_{\rho}^{\star\G}(p)$ as the largest confidence attainable by a Gaussian receiver, Def.~\ref{def:gdmax} gives $C_{\rho}^{\star\G}(p)=p\,{\rm e}^{D_{\max}^{\G}(\rho\|\sigma)}\big/\big[1-p+p\,{\rm e}^{D_{\max}^{\G}(\rho\|\sigma)}\big]$. The postselection is purely classical and therefore does not alter the Gaussian restriction on the physical receiver.

Maximum-confidence discrimination has also been demonstrated experimentally in an optical implementation~\cite{Mosley_2006}.}

\bibliography{bibshort}
%\linespread{1.0}\selectfont

\end{document}